\documentclass[12pt,english]{article}

\usepackage[T1]{fontenc}
\usepackage[margin=1in]{geometry}
\usepackage[latin9]{inputenc}
\usepackage{color}
\usepackage{mathtools}

\usepackage[usenames,dvipsnames]{xcolor}
\usepackage{amsfonts, amssymb, amsmath, amsthm}
\usepackage{hyperref}
\definecolor{darkblue}{rgb}{0, 0, 0.5}
\hypersetup{
     colorlinks = true,
     citecolor = darkblue,
     urlcolor = darkblue,
     linkcolor = darkblue
}
\usepackage[all]{hypcap}
\usepackage[nameinlink, capitalize, noabbrev]{cleveref}

\usepackage[dvipsnames]{xcolor}
\usepackage{listings}

\definecolor{codegreen}{rgb}{0,0.6,0}
\definecolor{codegray}{rgb}{0.5,0.5,0.5}
\definecolor{codepurple}{rgb}{0.58,0,0.82}
\definecolor{backcolour}{rgb}{0.95,0.95,0.92}
\definecolor{deepblue}{rgb}{0,0,.9}

\lstdefinestyle{mystyle}{
  language = Python,
  otherkeywords={@},
  deletekeywords={filter},
  deletendkeywords={max, min, abs},
  backgroundcolor=\color{white}, commentstyle=\color{codegreen},
  keywordstyle=\color{magenta},
  numberstyle=\tiny\color{codegray},
  stringstyle=\color{codepurple},
  basicstyle=\ttfamily\footnotesize,
  breakatwhitespace=false,         
  breaklines=true,                 
  captionpos=b,                    
  keepspaces=true,                 
  numbers=left,                    
  numbersep=5pt,                  
  showspaces=false,                
  showstringspaces=false,
  showtabs=false,                  
  tabsize=2,
  emph={zeros, prange, jit, power, std, rvs, maximum, roll, convolve, array, max, min, pdf, ppf, uniform, random, sqrt, linspace, sort, mean, percentile, sign, abs, powerlaw, beta, stats},
  emphstyle=\color{deepblue}
}

\usepackage{changes}
\definechangesauthor[name=Pasha, color=BrickRed]{P}
\definechangesauthor[name=Grisha, color=ForestGreen]{G}

\usepackage{enumerate}
\usepackage{enumitem} 
\usepackage{cases}
\usepackage{bm}
\usepackage[authoryear]{natbib}
\usepackage{setspace}
\usepackage{float}

\usepackage{graphicx, graphics}

\usepackage{booktabs}
\usepackage{rotating}
\usepackage[flushleft]{threeparttable}
\usepackage{caption, subcaption}
\usepackage{here, float}
\usepackage{tabularx, makecell}
\usepackage{adjustbox}
\usepackage{tabu}

\newcommand{\F}{\mathcal{F}}
\newcommand{\G}{\mathbb{G}}

\newcommand{\eG}{\mathbb{G}}

\newcommand{\R}{\mathbb{R}}
\newcommand{\E}{\mathbf{E}}
\newcommand{\Pb}{\mathbb{P}}
\newcommand{\Prb}{\mathbb{P}}

\newcommand{\bydef}{\coloneqq}

\newcommand{\bigPar}[1]{\left( #1 \right)}

\newcommand{\sumin}{\sum_{i=1}^n}

\makeatletter

\usepackage[ruled, vlined]{algorithm2e}

\newtheorem{assumption}{Assumption}
\crefname{assumption}{Assumption}{Assumptions}
\crefname{prop}{Proposition}{Propositions}
\crefname{lem}{Lemma}{Lemmas}
\crefname{thm}{Theorem}{Theorems}

\newtheorem{cor}{Corollary}
\newtheorem{lem}{Lemma}
\newtheorem{thm}{Theorem}

\newtheorem*{thm*}{Theorem}

\title{Bias correction and uniform inference\\for the quantile density function}
\author{
\setcounter{footnote}{1}
    Grigory Franguridi\thanks{
		Department of Economics, University of Southern California.
		Email: \href{franguri@usc.edu}{franguri@usc.edu}
	}
}
\date{\today}

\begin{document}
\maketitle
\begin{abstract}
    For the kernel estimator of the quantile density function (the derivative of the quantile function), I show how to perform the boundary bias correction, establish the rate of strong uniform consistency of the bias-corrected estimator, and construct the confidence bands that are asymptotically exact uniformly over the entire domain $[0,1]$. The proposed procedures rely on the pivotality of the studentized bias-corrected estimator and known anti-concentration properties of the Gaussian approximation for its supremum.
\end{abstract}

\section{Introduction}

The derivative of the quantile function, the \emph{quantile density} (QD), has been long recognized as an important object in statistical inference.\footnote{This function is sometimes also called the \emph{sparsity function} \citep{tukey1965part}.} In particular, it arises as a factor in the asymptotically linear expansion for the quantile function \citep{bahadur1966note,kiefer1967bahadur}, and hence may be used for asymptotically valid inference on quantiles \citep{csorgo1981strong,csorgo1981two,koenker_2005}.

Given its importance, several estimators of the QD have been proposed in the literature. The most widely used estimator is the \emph{kernel quantile density} (KQD), originally developed by \citet{siddiqui1960distribution} and \citet{bloch1968simple} for the case of rectangular kernel, and generalized to arbitrary kernels by \citet{falk1986estimation}, \citet{welsh1988asymptotically}, \citet{csorgHo1991estimating}, and \cite{jones1992estimating}. This estimator is simply a smoothed derivative of the empirical quantile function, where smoothing is performed via convolution with a kernel function.

Similarly to the classical case of kernel density estimation, the KQD suffers from bias close to the boundary points $\{0,1\}$ of its domain $[0,1]$, rendering the estimator inconsistent. To the best of my knowledge, no bias correction procedures have been developed for the QD.

In this paper, I show how to perform correction for the boundary bias, recovering strong uniform consistency for the resulting bias-corrected KQD (BC-KQD) estimator. The bias correction is computationally cheap and is based on the fact that the bias of the KQD is approximately equal to the integral of the localized kernel function, a quantity that only depends on the chosen kernel and bandwidth. I also develop an algorithm for construction of the uniform confidence bands around the QD on its entire domain $[0,1]$. This procedure relies on the fact that the studentized BC-KQD exhibits an influence function that is \emph{pivotal}. This makes it possible to calculate the critical values by simulating from either the known influence function or the studentized BC-KQD under an alternative (pseudo) distribution of the data.

The rest of the paper is organized as follows. \cref{sec:KQD} outlines the framework and defines the KQD estimator. \cref{sec:BC} introduces the BC-KQD estimator and establishes its Bahadur-Kiefer expansion. \cref{sec:CB} develops the uniform confidence bands based on the BC-KQD. \cref{sec:MC} illustrates the performance of the confidence bands in a set of Monte Carlo simulations. \cref{sec:conclusion} concludes. Proofs of theoretical results are given in the Appendix.

\section{Setup and kernel quantile density estimator}\label{sec:KQD}

The data consist of independent identically distributed draws $X_1,\dots,X_n$ from a distribution on $\R$ with a cumulative distribution function (CDF) $F$ satisfying the following assumption.

\begin{assumption}[Data generating process]\label{ass:dgp} The distribution $F$ has compact support $[\underline x, \bar x]$ and admits a density $f=F'$ that is continuously differentiable and bounded away from zero and infinity on $[\underline x, \bar x]$.
\end{assumption}
Assumption \ref{ass:dgp} implies that the quantile density
\begin{align}
    q(u) \bydef \frac{d F^{-1}(u)}{du} = \frac{1}{f(F^{-1}(u))}
\end{align}
is continuously differentiable and bounded away from zero and infinity on the support $[\underline x, \bar x]$.

Let $X_{(1)}\le \cdots \le X_{(n)}$ be the order statistics of the sample $X_1,\dots,X_n$, and let $\hat Q$ denote the empirical quantile function,
\begin{align}
    \hat Q(u) \bydef
    \begin{cases}
    X_{(\left\lfloor{nu}\right\rfloor +1)}, \quad u\in [0,1),\\
    X_{(n)}, \quad u=1,
    \end{cases}
\end{align}
The KQD estimator is defined as
\begin{align}
    \hat q_h(u) &\bydef \int_{0}^{1} K_h(u-z) \, d\hat Q(z) = \sum_{i=1}^{n-1} K_h\left(u-\frac{i}{n}\right) \bigPar{X_{(i+1)}-X_{(i)}}, \quad u\in [0,1],
\end{align}
where $K$ is a kernel function, $K_h(z) \bydef h^{-1}K\left(h^{-1}z\right)$, and $h>0$ is bandwidth \citep[see, e.g.,][]{csorgHo1991estimating}. We impose the following assumptions on the kernel and bandwidth.

\begin{assumption}[Kernel function]\label{ass:kernel}
The kernel $K$ is a nonnegative function of bounded variation that is supported on $[-1/2,1/2]$, symmetric around $0$, and satisfies
\begin{align}
    \int_\R K(x) \, dx = 1, \quad \int_\R K^2(x) \, dx < \infty.
\end{align}
\end{assumption}

\begin{assumption}[Bandwidth, estimation]\label{ass:bandwidth-for-consistency}
The bandwidth $h=h_n$ is such that $h_n \to 0$ and
\begin{enumerate}
    \item \label{ass:band-large-cons} $h_n^{-1} = o\left(n^{1/2} (\log n)^{-1} (\log\log n)^{-1/2}\log h^{-1} \right)$,
    \item \label{ass:band-undersmooth-cons} $h_n=o\left(n^{-1/3}(\log h^{-1})^{-1/3}\right)$.
\end{enumerate}
\end{assumption}

\begin{assumption}[Bandwidth, inference]\label{ass:bandwidth}
The bandwidth $h=h_n$ is such that $h_n \to 0$ and
\begin{enumerate}
    \item \label{ass:band-large} $h_n^{-1} = o\left(n^{1/2} (\log n)^{-2} (\log\log n)^{-1/2} \right)$,
    \item \label{ass:band-undersmooth} $h_n=o\left(n^{-1/3}(\log n)^{-1}\right)$.
\end{enumerate}
\end{assumption}
Assumption \ref{ass:kernel} is standard; boundedness of the total variation of $K$ ensures that the class
\begin{align}
    \mathcal{F} \bydef \left\{K \left(\frac{u-\cdot}{h}\right), \,\,u \in [0,1], h >0 \right\}
\end{align}
is a bounded VC class of measurable functions, see, e.g., \citet{nolan1987u}.

Assumptions \ref{ass:bandwidth-for-consistency} and \ref{ass:bandwidth} are essentially the same, up to the log terms in the bandwidth rates, with Assumption \ref{ass:bandwidth-for-consistency} being slightly weaker. Assumption \ref{ass:bandwidth}.\ref{ass:band-large} states that the bandwidth rate is large enough (slightly larger than $n^{-1/2}$) to guarantee that the smoothed remainder of the classical Bahadur-Kiefer expansion vanishes asymptotically, see the proof of \cref{cor:conv-rate} below. Assumption \ref{ass:bandwidth}.\ref{ass:band-undersmooth} imposes the \emph{undersmoothing} bandwidth rate (slightly smaller than $n^{-1/3}$), which ensures that the smoothing bias disappears fast enough for the confidence bands to be valid, see the proof of \cref{thm:CB} below.

\section{Bias correction and Bahadur-Kiefer expansion}\label{sec:BC}

In this section, I introduce the bias-corrected estimator and develop its asymptotically linear expansion with an explicit a.s. uniform rate of the remainder (the Bahadur-Kiefer expansion).

To see the necessity of bias correction, note that, for $u$ close to the boundary, the kernel weights $K_h(u-i/n)$, $i=1,\dots,n-1$, do not approximately sum up to one, rendering the KQD $\hat q_h(u)$ inconsistent. Therefore, dividing the KQD by the sum of the kernel weights (or the corresponding integral of the kernel function) may eliminate the boundary bias. To this end, define
\begin{align}
    \psi_h(u) &\bydef \int_0^1 K_h(u-z) \, dz = \int_{\max(u-h/2,0)}^{\min(u+h/2,1)} K_h(u-z) \, dz, \quad u\in[0,1].
\end{align}
For computational purposes, note that $\psi_h$ is symmetric around $1/2$ (i.e. $\psi_h(u) = \psi_h(1-u)$ for all $u \in [0,1]$), $\psi_h\in[1/2,1]$ and $\psi_h(u) = 1$ for $u \in [h/2,1-h/2]$.
The bias-corrected KQD (BC-KQD) is then defined as
\begin{align}
    \hat q_h^{bc}(u) \bydef \frac{\hat q_h(u)}{\psi_h(u)} = \frac{\sum_{i=1}^{n-1} K_h\left(u-\frac{i}{n}\right) \bigPar{X_{(i+1)}-X_{(i)}}}{\int_0^1 K_h(u-z) \, dz}, \quad u\in [0,1].
\end{align}

The following theorem establishes that the studentized BC-KQD is approximately equal to the centered kernel density estimator with an approximation error that converges to zero a.s. at an explicit uniform rate. Since this result resembles (and relies on) the classical asymptotically linear expansion for the quantile function \citep{bahadur1966note,kiefer1967bahadur}, we call it the \emph{Bahadur-Kiefer expansion for the BC-KQD}. Denote $U_i=F(X_i)$, $i=1,\dots,n.$

\begin{thm}[Bahadur-Kiefer expansion for the BC-KQD]\label{Thm:BK-expansion}
Suppose Assumptions \ref{ass:dgp} and \ref{ass:kernel} are satisfied and $h_n \to 0$. Then the following representation holds uniformly in $u \in [0,1]$,
\begin{align}
    Z_n^{bc}(u) = -\G_n(u) + O_{a.s.}\left( n^{1/2} h^{3/2} + h \log h^{-1} + h^{-1/2}n^{-1/4}(\log n)^{1/2}(\log \log n)^{1/4} \right),
\end{align}
where
\begin{align}
    Z_n^{bc}(u) &\bydef \frac{\sqrt{nh}\left( \hat q_h^{bc}(u)-q(u) \right)}{q(u) / \psi_h(u)},\\
    \G_n(u) &\bydef \frac{1}{\sqrt{nh}}\sumin\left[K\bigPar{\frac{U_i-u}{h}}-\E K\bigPar{\frac{U_i-u}{h}} \right] \\
    &= \sqrt{nh} \cdot \frac{1}{n}\sumin \left[K_h(U_i-u) - \psi_h(u)\right].
\end{align}
\end{thm}

This representation allows us to establish the exact rate of strong uniform consistency of the BC-KQD under a bandwidth that achieves undersmoothing (Assumption \ref{ass:bandwidth-for-consistency}.\ref{ass:band-undersmooth-cons}).

\begin{cor}[Strong uniform consistency of BC-KQD]\label{cor:conv-rate}
Suppose Assumptions \ref{ass:dgp}, \ref{ass:kernel}, and \ref{ass:bandwidth-for-consistency} hold. Then
\begin{align}
    \lim_{n \to \infty} \sqrt{\frac{nh_n}{2\log h_n^{-1}}} \sup_{u \in [0,1]} \left| \hat q_h^{bc}(u) - q(u) \right| = \left(\int_\R K^2(x) \, dx \right)^{1/2} \text{ a.s.}
\end{align}
\end{cor}

One of the convenient features of the KQD (and BC-KQD) estimator is that its bandwidth has a natural scale $[0,1]$ which is independent of the data generating process. Hence, I put aside the choice of constant $c$ in the bandwidth $h=c n^{-\eta}$ and suggest setting $c=1$.

Regarding the choice of the rate $\eta$, ignoring the log terms, it is easy to establish the rate-optimal bandwidth, which is achieved whenever the rate of the smoothing bias $n^{1/2}h^{3/2}$ matches that of the remainder in the original Bahadur-Kiefer expansion $n^{-1/4}h^{-1/2}$. It follows that the nearly-optimal bandwidth is
\begin{align}
    h_n^{opt} = O\left(n^{-3/8}\right).
\end{align}
Under this bandwidth, the exact rate of strong uniform convergence is
\begin{align}
    O\left(\frac{\log n}{n^{5/16}}\right),
\end{align}
which is just slightly worse than the familiar ``cube-root'' rate \citep{kim1990cube}.

\section{Uniform confidence bands}\label{sec:CB}

Suppose we had access to valid approximations $c_{n,\tau}$, $c_{n,\tau}^{abs}$ to the $\tau$-quantiles of the random variables
\begin{align}
    W_n^{bc} &= \sup_{u\in[0,1]} Z_n^{bc}(u),\\
    W_n^{bc,abs} &= \sup_{u\in[0,1]}\left| Z_n^{bc}(u) \right|,
\end{align}
respectively, in the sense that
\begin{align}
    \Prb(W_n^{bc} \le c_{n,\tau}) &= \tau + o(1), \label{eq:valid-cv-abstract}\\
    \Prb(W_n^{bc,abs} \le c_{n,\tau}^{abs}) &= \tau + o(1).
\end{align}
Then the following confidence bands for $q(\cdot)$ would be asymptotically valid at the confidence level $1-\alpha$:
\begin{enumerate}
    \item the one-sided CB
        \begin{align}
            \left[ \frac{\hat q_h^{bc}(u)}{1 + \frac{c_{n,1-\alpha}}{\psi_h(u)\sqrt{nh}}}, \, +\infty\right), \quad u \in [0,1], \label{eq:CB-1}
        \end{align}
    \item the one-sided CB
        \begin{align}
            \left(-\infty, \, \frac{\hat q_h^{bc}(u)}{1 - \frac{c_{n,1-\alpha}}{\psi_h(u)\sqrt{nh}}} \right], \quad u \in [0,1], \label{eq:CB-2}
        \end{align}
    \item the two-sided CB
        \begin{align}
            q(u) \in \left[\frac{\hat q_h^{bc}(u)}{1 + \frac{c_{n,1-\alpha/2}^{abs}}{\psi_h(u)\sqrt{nh}}}, \,\, \frac{\hat q_h^{bc}(u)}{1 - \frac{c_{n,1-\alpha/2}^{abs}}{\psi_h(u)\sqrt{nh}}} \right], \quad u \in [0,1]. \label{eq:CB-3}
        \end{align}
\end{enumerate}

I propose two ways of obtaining such approximate critical values, both making use of the pivotality of the studentized bias-corrected KQD $Z_n^{bc}(u)$, see \cref{Thm:BK-expansion}. I focus on the one-sided critical value $c_{n,\tau}$ for simplicity; the proofs for the two-sided critical value are analogous.

The first approach is to let $c_{n,\tau}$ be the $\tau$-quantile of the random variable
\begin{align}
    W_n^\G = \sup_{u \in [0,1]} \G_n(u).
\end{align}
Since $\G_n$ is a known process, $c_{n,\tau}$ can be obtained easily by simulation. In principle, $c_{n,\tau}$ can be tabulated for different choices of the kernel $K$ and values of the sample size $n$ and the bandwidth $h$.

The other approach is to let $c_{n,\tau}$ be the $\tau$-quantile of the random variable
\begin{align}
    W_n^{U[0,1]} \bydef \sup_{u\in[0,1]} Z_n^{bc,U[0,1]}(u),
\end{align}
 where $Z_n^{bc,U[0,1]}(u)$ is equal to $Z_n^{bc}(u)$ evaluated at a pseudo-sample $\tilde X_1,\dots,\tilde X_n \sim U[0,1]$ in place of the original sample. For the uniform distribution, $q \equiv 1$, and hence
\begin{align}
    Z_n^{bc,U[0,1]}(u) \bydef \frac{\sqrt{nh}(\tilde q_h^{bc}(u)-q(u))}{q(u)/\psi_h(u)} = \sqrt{nh}(\tilde q_h(u)-\psi_h(u)),
\end{align}
where $\tilde q_n(u)$ is the (non-bias-corrected) KQD calculated using the pseudo-sample, i.e.
\begin{align}
    \tilde q_n(u) = \sum_{i=1}^{n-1} K_h\left(u-\frac{i}{n}\right) \bigPar{\tilde X_{(i+1)}-\tilde X_{(i)}}, \quad u\in [0,1].
\end{align}

The following theorem establishes that the two aforementioned approximations to the critical values are valid, implying the asymptotic validity of the confidence bands. These confidence bands are centered at an AMSE-suboptimal estimator $\hat q_h^{bc}$ and are expected to shrink at a rate slightly slower than the minimax optimal rate, as noted by \citet[p.1795]{chernozhukov2014anti}. This is compensated for by the confidence bands exhibiting the coverage that is asymptotically \emph{exact}.

\begin{thm}[Exactness of confidence bands]\label{thm:CB}
Suppose Assumptions \ref{ass:dgp}, \ref{ass:kernel}, and \ref{ass:bandwidth} hold. Then
\begin{align}
     &\lim_{n \to \infty} \sup_{t\in\R}\left| \Prb\left(W_n^{bc} \le t \right) - \Prb\left(W_n^{\G} \le t \right) \right| = 0,\\
     &\lim_{n \to \infty} \sup_{t\in\R}\left| \Prb\left(W_n^{bc} \le t \right) - \Prb\left(W_n^{U[0,1]} \le t \right) \right| = 0,
\end{align}
and hence the confidence bands \eqref{eq:CB-1}, \eqref{eq:CB-2}, and \eqref{eq:CB-3} are asymptotically exact.
\end{thm}

\section{Monte Carlo study}\label{sec:MC}

In this section I study the finite-sample behavior of the proposed confidence bands in a set of Monte Carlo simulations.

I consider the following distributions of the data, all supported on the interval $[0,1]$: (i) uniform[0,1] distribution (ii) the distribution $N(1/2,1)$ truncated to $[0,1]$ (iii) the linear distribution with the PDF $f(x) = x+1/2$, $x\in[0,1]$. I set the nominal confidence level to be $1-\alpha \in \{0.8,0.9,0.95,0.99\}$ and the sample size $n \in \{100,500,1000,5000\}$. The critical values are obtained by simulating $\G_n(u)$ and calculating the quantiles of its supremum on the grid $u \in \{0.005,0.015,0.02,\dots,0.995\}$, with the number of simulations set to $20000$ (simulation results for the critical values based on $Z_n^{bc,U[0,1]}(u)$ are very similar, so I do not report them here). I use the kernel corresponding to the standard normal distribution truncated to $[-1/2,1/2]$ and the nearly-optimal bandwidth $h=cn^{-3/8}$, where I set $c=1$ since the scale of the bandwidth is $[0,1]$, see \cref{sec:BC}.

In \cref{fig:CB}, included for illustration, I plot 100 independent realizations of the $90\%$ confidence bands for the linear distribution, along with the true quantile density (in blue). \cref{tab:coverage} contains simulated coverage values for the two-sided confidence bands. The coverage is almost invariant to the distribution of the data, but the size distortion tends to be smaller for higher nominal confidence levels.

\begin{figure}[t!]
\centering
\includegraphics[scale=0.5]{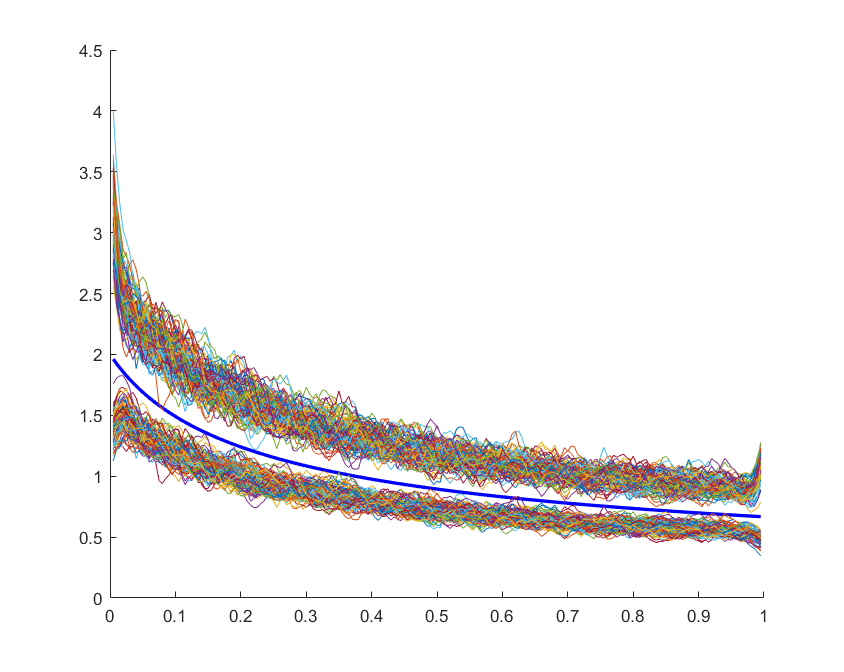}
\caption{$90\%$ confidence bands for the quantile density (in blue) of the linear distribution with the PDF $f=x+0.5$, $x\in [0,1]$. Number of independent realizations of the bands $S=100$, sample size $n=5000$.}
\label{fig:CB}
\end{figure}

\begin{table}[t!]
    \centering
    \begin{tabular}{l c c c c}
        Confidence level & $0.8$ & $0.9$ & $0.95$ & $0.99$ \\
        \hline
        & \multicolumn{4}{c}{Uniform distribution} \\
        \hline
        $n=100$& 0.891&	0.936&	0.962&	0.986 \\
        $n=500$& 0.881&	0.943&	0.966&	0.990 \\
        $n=1000$& 0.898&	0.947&	0.970&	0.993 \\
        $n=5000$& 0.907&	0.949&	0.976&	0.996 \\
        \hline
        & \multicolumn{4}{c}{Linear distribution} \\
        \hline
        $n=100$& 0.891&	0.929&	0.956&	0.987 \\
        $n=500$& 0.878&	0.936&	0.961&	0.989 \\
        $n=1000$& 0.890&	0.944&	0.970&	0.991 \\
        $n=5000$& 0.914&	0.949&	0.976&	0.996 \\
        \hline 
        & \multicolumn{4}{c}{Truncated normal distribution} \\
        \hline
        $n=100$& 0.898&	0.942&	0.964&	0.988 \\
        $n=500$& 0.887&	0.944&	0.967&	0.992 \\
        $n=1000$& 0.905&	0.950&	0.972&	0.993 \\
        $n=5000$& 0.911&	0.952&	0.978&	0.997 \\
    \end{tabular}
    \caption{Simulated coverage of the two-sided confidence bands}
    \label{tab:coverage}
\end{table}

\section{Conclusion}\label{sec:conclusion}

To the best of my knowledge, no boundary bias correction or uniform inference procedures have been developed for the quantile density (sparsity) function.
In this paper, I develop such procedures, establish their validity and show in a set of Monte Carlo simulations that they perform reasonably well in finite samples. I hope that, even when the quantile density itself is not the main inference target, these results may be employed for improving the quality of inference for other statistical objects, including the quantile function.


%

\bibliographystyle{ecta}
\bibliography{main}
 
\newpage  
\part*{Appendix}
\appendix
 
\setcounter{page}{1} 

\section{Proof of \cref{Thm:BK-expansion} and \cref{cor:conv-rate}}
First, note that
\begin{align}
    q_h(u) &\bydef \int_0^1 K_h(u-z) q(z) \, dz = \int_0^1 K_h(u-z) \left( q(u) + q'(\xi(u,z))(z-u) \right) \, dz \\
    &= q(u) \psi_h(u) + r_h(u),
\end{align}
where $r_n(u) = O(h)$ uniformly in $u\in[0,1]$ since $q$ is continuously differentiable on $[0,1]$.

Therefore,
\begin{align}
    Z_n^{bc}(u) &\bydef \frac{\sqrt{nh}\left( \hat q_h^{bc}(u)-q(u) \right)}{q(u) / \psi_h(u)} =  \frac{\sqrt{nh}\left( \hat q_h(u)- \psi_h(u) q(u) \right)}{q(u)} = Z_n^c(u) + r_n^{bc}(u),
\end{align}
where 
\begin{align}
    Z_n^c(u) &\bydef \frac{\sqrt{nh}\left( \hat q_h(u)-q_h(u) \right)}{q(u)},\\
    r_n^{bc}(u) &= \frac{\sqrt{nh} r_h(u)}{q(u)} = O\left( n^{1/2}h^{3/2}\right) \text{ uniformly in } u\in[0,1].
\end{align}

The result now follows from the asymptotically linear expansion of the process $Z_n^c$,
\begin{align}
    Z_n^c(u) = \G_n(u) + O_{a.s.}\left(h \log h^{-1} + h^{-1/2}n^{-1/4}(\log n)^{1/2}(\log \log n)^{1/4}\right), \label{eq:Znc-expansion}
\end{align}
This expansion is implied by the proof of \citet[Theorem 1]{franguridi2022}. I reproduce this proof here for completeness.

\subsection{Proof of the representation \eqref{eq:Znc-expansion}}

First, we need the following two lemmas concerning expressions that appear further in the proof.

\begin{lem}\label{Lem:int-by-parts-kqe}
Suppose that Assumptions \ref{ass:dgp} and \ref{ass:kernel} hold. Then, for every $u \in [0,1]$,
\begin{align}
    \int_0^1 K_h(u-z) \, d\left(\hat Q(z)-Q(z)\right) = - \int_0^1 \left(\hat Q(z)-Q(z)\right)  \, d K_h(u-z) + R^I_n(u),
\end{align}
where $\sup_{u\in [0,1]} |R^I_n(u)| = O_{a.s.}\left(\frac{1}{nh}\right)$.
\end{lem}

\begin{proof}
Denote $\hat\psi(z)=\hat Q(z)-Q(z)$ and note that $\hat \psi$ is a function of bounded variation a.s. Using integration by parts for the Riemann-Stieltjes integral \citep[see e.g.][Theorem 1.2.7]{stroock1998concise}, we have
\begin{align}
    \int_0^1 K_h(u-z) \, d\hat \psi(z) = - \int_0^1 \hat\psi(z)  \, d K_h(u-z) + K_h(u-1) \hat\psi(1) - K_h(u)\hat\psi(0)
\end{align}
To complete the proof, note that $\hat\psi(1)=X_{(n)}- \bar x = O_{a.s.}(n^{-1})$, $\hat\psi(0)=X_{(1)}-\underline x = O_{a.s.}(n^{-1})$, $|K_h(u-1)|\le h^{-1}K(0)$ and $|K_h(u)|\le h^{-1}K(0)$.
\end{proof}

\begin{lem}\label{Lem:int-by-parts-kde}
Suppose that Assumptions \ref{ass:dgp} and \ref{ass:kernel} hold. Then, for every $u \in [0,1]$,
\begin{align}
    \int_0^1 (\hat F(Q(z))-z) \, d K_h(u-z) &= -\G_n(u)/\sqrt{nh}.
\end{align}
\end{lem}
\begin{proof}
Using integration by parts for the Riemann-Stieltjes integral \citep[see e.g.][Theorem 1.2.7]{stroock1998concise}, we have
\small
\begin{align}
    \int_0^1 (\hat F(Q(z))-z) \, d K_h(u-z) &= -\int_0^1  K_h(u-z) \, d \left[\hat F(Q(z))-z\right] + K_h(u-1) \left[ \hat F(\bar x)-1 \right] + K_h(u)\hat F(\underline x) \\
    &= -\int_0^1  K_h(u-z) \, d \left[\hat F(Q(z))-z\right],
\end{align}
\normalsize
where we used the fact that $\hat F(\bar x) = 1$ a.s. and $\hat F(\underline x)=0$ a.s. We further write
\begin{align}
        \int_0^1(\hat F(Q(z))-z) \, d K_h(u-z) &= -\int_0^1 K_h(u-z) \, d \left[\hat F(Q(z))-z\right] \\
        &= -\int_0^{\bar b} K_h(u-F(x)) \, d \left[\hat F(x)-F(x) \right] \\
        &= -\frac{1}{n} \sumin \left[K_h(u-F(b_i)) - \E K_h(u-F(b_i)) \right]\\
        &=: -\G_n(u)/\sqrt{nh},
\end{align}
where in the second equality we used the change of variables $x=Q(z)$.
\end{proof}


We now proceed with the proof of representation \eqref{eq:Znc-expansion}.

Recall the classical Bahadur-Kiefer expansion \citep{bahadur1966note,kiefer1967bahadur},
\begin{align}\label{E:BK-general}
 \hat Q(u)-Q(u) &= - q(u)\left( \hat F(Q(u)) - u \right) + r_n(u),\\
 \text{where } r_n(u) &= O_{a.s.}\left( n^{-3/4} \ell(n) \right) \text{ uniformly in } u \in [0,1],
\end{align}
and $\ell(n) \bydef (\log n)^{1/2}(\log \log n)^{1/4}$. Combine this expansion with Lemma \ref{Lem:int-by-parts-kqe} to obtain
\begin{align}
\hat q_h(u) - q_h(u) &= \int_0^1 K_h(u-z) \, d\left[ \hat Q(z)-Q(z)\right] \\
&= \int_0^1 \left[ \hat Q(z)-Q(z)\right] \, d K_h(u - z) + R_n^I(u) \\
&= \int_0^1  q(z) (\hat F(Q(z))-z)\, d K_h(u - z) + \int_0^1  R_n^{BK}(z) \, d K_h(u - z) + R_n^I(u). \label{E:qh-q}
\end{align}

\textbf{First term in \eqref{E:qh-q}}.

Since $f$ is bounded away from zero, $|q'|\le M< \infty$ for some constant $M$, and hence $|q(z)-q(u)|\le M|z-u|$. The first term in \eqref{E:qh-q} can then be rewritten as 
\begin{align}
        \int_0^1q(z) (\hat F(Q(z))-z) \, d K_h(u-z) &= q(u) \int_0^1(\hat F(Q(z))-z) \, d K_h(u-z) + R^{II}_n(u), \label{E:first-term-split}
\end{align}
where
\begin{align}
    \left|R_n^{II}(u)\right| &= \left| \int_0^1(q(z)-q(u))  (\hat F(Q(z))-z)\, d K_h(u - z) \right| \\
    &\le Mh \left| \int_0^1(\hat F(Q(z))-z) \, d K_h(u-z) \right| = Mh\left|\G_n(u)/\sqrt{nh}\right|,
\end{align}
the last equality using Lemma \ref{Lem:int-by-parts-kde}. The process $\G_n$ has the strong uniform convergence rate $\log h^{-1} / \sqrt{nh}$ \citep[see, e.g.,][]{gine2002rates}, and hence
\begin{align}
    R_n^{II}(u) = O_{a.s.}\left(\frac{h\log h^{-1}}{\sqrt{nh}}\right) \text{ uniformly over } u \in[0,1].
\end{align}
Applying Lemma \ref{Lem:int-by-parts-kde} to the first term in \eqref{E:first-term-split} allows us to rewrite
\begin{align}
    \int_0^1 q(z) (\hat F(Q(z))-z) \, d K_h(u-z) &= -q(u) \frac{\G_n(u)}{\sqrt{nh}} + O_{a.s.}\left(\frac{h\log h^{-1}}{\sqrt{nh}}\right). \label{E:first-term-rate}
\end{align}

\textbf{Second term in \eqref{E:qh-q}}.

This term can be upper bounded as follows,
\begin{align}
        &\sup_u \left|\int_0^1R_n^{BK}(z) \, d K_h(u-z) \right| \le \sup_u \int_0^1\left|  R_n^{BK}(z) \right| \,\left| d K_h(u-z)\right| \le \sup_z |R_n^{BK}(z)| TV(K_h)\\
        &= O_{a.s.}\left(n^{-3/4}\ell(n)\right) h^{-1}TV(K) =
         O_{a.s.}\left(h^{-1}n^{-3/4}\ell(n)\right), \label{E:second-term-rate}
\end{align}
where we used the properties of total variation in the first inequality and in the second equality.

Plugging \eqref{E:first-term-rate} and \eqref{E:second-term-rate} into \eqref{E:qh-q} and multiplying by $\sqrt{nh}$ yields
\begin{align}
    \sqrt{nh}\left(\hat q_h(u)-q_h(u) \right) = -q(u)\G_n(u) + O_{a.s.}(h\log h^{-1}) + O_{a.s.}\left(h^{-1/2}n^{-1/4}\ell(n)\right). \label{eq:qhat-q}
\end{align}

Note that we disregarded the term $\sqrt{nh} R_n^I(u)$, since it has the uniform order $O_{a.s.}(n^{-1/2}h^{-1/2})$, which is smaller than  $O_{a.s.}\left(h^{-1/2}n^{-1/4}\ell(n)\right)$. Dividing by $q(u)$, which is bounded away from zero for $u \in [0,1]$ due to Assumption \ref{ass:dgp}, finishes the proof. $\qed$

\subsection{Proof of \cref{cor:conv-rate}}

Let us check that the conditions of \citet[Proposition 3.1]{gine2002rates} hold. Indeed, Assumption \ref{ass:kernel} implies their condition $(K_2)$, while Assumption \ref{ass:bandwidth-for-consistency} implies their conditions (2.11) and $(W_2)$. By \citet[Remark 3.5]{gine2002rates}, their condition $(D_2)$ can be replaced by the conditions satisfied by the uniform distribution. To complete the proof, divide the expansion in \cref{Thm:BK-expansion} by $\sqrt{2\log h_n^{-1}}$ and note that the first term $\G_n(u)/\sqrt{2\log h_n^{-1}}$ converges to $\left(\int_\R K^2(x)\,dx\right)^{1/2}$ by \citet[Proposition 3.1]{gine2002rates}, while the remainder converges to zero a.s. due to Assumption \ref{ass:bandwidth-for-consistency}. $\qed$

\section{Proof of \cref{thm:CB}}

A key ingredient of the proof is to note that Lemmas 2.3 and 2.4 of \citet{chernozhukov2014gaussian} continue to hold even if their random variable $Z_n$ does not have the form $Z_n=\sup_{f\in \mathcal{F}_n} \eG_nf$ for the standard empirical process $\eG_n$, but instead is a generic random variable admitting a strong sup-Gaussian approximation with a sufficiently small remainder.

For completeness, we provide the aforementioned trivial extensions of the two lemmas here, taken directly from \citet{franguridi2022}.

Let $X$ be a random variable with distribution $P$ taking values in a measurable space $(S,\mathcal{S})$. Let $\F$ be a class of real-valued functions on $S$. We say that a function $F: S \to \R$ is an \emph{envelope} of $\F$ if $F$ is measurable and $|f(x)|\le F(x)$ for all $f \in \F$ and $x \in S$.

We impose the following assumptions (A1)-(A3) of \citet{chernozhukov2014gaussian}.
\begin{enumerate}
    \item[(A1)] The class $\F$ is \emph{pointwise measurable}, i.e. it contains a coutable subset $\G$ such that for every $f\in \F$ there exists a sequence $g_m \in \G$ with $g_m(x) \to f(x)$ for every $x \in S$.
    \item[(A2)] For some $q\ge 2$, an envelope $F$ of $\F$ satisfies $F \in L^q(P)$.
    \item[(A3)] The class $\F$ is $P$-pre-Gaussian, i.e. there exists a tight Gaussian random variable $G_P$ in $l^\infty(\F)$ with mean zero and covariance function
    \begin{align}
        \E[G_P(f)G_P(g)] = \E[f(X)g(X)] \text{ for all } f,g\in\F.
    \end{align}
\end{enumerate}

\begin{lem}[A trivial extension of Lemma 2.3 of \citet{chernozhukov2014gaussian}] \label{Lem:CCK-2-3}
Suppose that Assumptions (A1)-(A3) are satisfied and that there exist constants $\underline \sigma$, $\bar \sigma>0$ such that $\underline\sigma^2 \le Pf^2 \le \bar\sigma^2$ for all $f\in\F$. Moreover, suppose there exist constants $r_1,r_2>0$ and a random variable $\tilde Z=\sup_{f\in \F} G_Pf$ such that $\Pb(|Z-\tilde Z|>  r_1)\le r_2$. Then
\begin{align}
    \sup_{t\in\R}\left|\Pb(Z\le t)-\Pb(\tilde Z\le t)\right| \le C_{\sigma} r_1\left\{ \E \tilde Z + \sqrt{1 \vee \log(\underline\sigma/r_1)} \right\} + r_2,
\end{align}
where $C_\sigma$ is a constant depending only on $\underline\sigma$ and $\bar\sigma$.
\end{lem}

\begin{proof}
For every $t\in \R$, we have
\begin{align}
    \Pb(Z\le t) &= \Pb(\{Z\le t\} \cap \{|Z-\tilde Z|\le r_1\}) + \Pb(\{Z\le t\} \cap \{|Z-\tilde Z|>r_1\})\\
    &\le \Pb(\tilde Z\le t+r_1)+r_2\\
    &\le \Pb(\tilde Z\le t) + C_{\sigma} r_1\left\{ \E \tilde Z + \sqrt{1 \vee \log(\underline\sigma/r_1)} \right\} + r_2,
\end{align}
where Lemma A.1 of \citet{chernozhukov2014gaussian} (an anti-concentration inequality for $\tilde Z$) is used to deduce the last inequality. A similar argument leads to the reverse inequality, which completes the proof.
\end{proof}

\begin{lem}[A trivial extension of Lemma 2.4 of \citet{chernozhukov2014gaussian}] \label{Lem:CCK-2-4}
Suppose that there exists a sequence of $P$-centered classes $\F_n$ of measurable functions $S \to \R$ satisfying assumptions (A1)-(A3) with $\F=\F_n$ for each $n$, where in the assumption (A3) the constants $\underline\sigma$ and $\bar\sigma$ do not depend on $n$. Denote by $B_n$ the Brownian bridge on $\ell^\infty(\F_n)$, i.e. a tight Gaussian random variable in $\ell^\infty(\F_n)$ with mean zero and covariance function
\begin{align}
    \E[B_n(f)B_n(g)] = \E[f(X)g(X)] \text{ for all } f,g\in\F_n.
\end{align}
Moreover, suppose that there exists a sequence of random variables $\tilde Z_n = \sup_{f\in\F_n} B_n(f)$ and a sequence of constants $r_n\to 0$ such that $|Z_n-\tilde Z_n|=O_P(r_n)$ and $r_n \E \tilde Z_n \to 0$. Then 
\begin{align}
    \sup_{t\in \R} \left|\Pb(Z_n\le t)-\Pb(\tilde Z_n\le t)\right| \to 0.
\end{align}
\end{lem}

\begin{proof}
Take $\beta_n\to\infty$ sufficiently slowly such that $\beta_nr_n(1\vee \E\tilde Z_n)=o(1)$. Then since $\Pb(|Z_n-\tilde Z_n|>\beta_n r_n)=o(1)$, by Lemma \ref{Lem:CCK-2-3}, we have
\begin{align}
    \sup_{t\in \R} \left|\Pb(Z_n\le t)-\Pb(\tilde Z_n\le t)\right| = O\left( r_n(\E\tilde Z_n + |\log (\beta_n r_n)|) \right) + o(1) = o(1).
\end{align}
This completes the proof.
\end{proof}

I now go back to the proof of \cref{thm:CB}. \citet[][Proposition 3.1]{chernozhukov2014gaussian} establish a sup-Gaussian approximation of $W_n^\G$; namely, there exists a tight centered Gaussian random variable $B_n$ in $\ell^\infty([0,1])$ with the covariance function
\begin{align}
\E[B_n(u)B_n(v)] = \text{Cov}\left(K_{h}(U-u), K_{h}(U-v) \right), \quad u,v\in[0,1],
\end{align}
where $U \sim \text{Uniform}[0,1]$, such that, for $\tilde W_n \bydef \sup_{u\in[0,1]} B_n(u)$, we have the approximation
\begin{align}
    W_n^\G = \tilde W_n + O_{p}\left((nh)^{-1/6}\log n\right) \label{eq:WG-Gaussian}.
\end{align}

\cref{Lem:CCK-2-4} and \citet[][Remark 3.2]{chernozhukov2014gaussian} then imply
\begin{align}
    \sup_{t\in \R} \left|\Pb(W_n^\G\le t)-\Pb(\tilde W_n\le t)\right| \to 0. \label{E:Kolm-conv-1}
\end{align}
On the other hand, from \cref{Thm:BK-expansion} it follows that
\begin{align}
    W_n^{bc} = W_n^\G + O_{a.s.}\left( n^{1/2} h^{3/2} + h \log h^{-1} + h^{-1/2}n^{-1/4}\ell(n) \right), \label{eq:Wbc-WG}
\end{align}
where we define $\ell(n) \bydef (\log n)^{1/2}(\log \log n)^{1/4}$.
Substituting \eqref{eq:WG-Gaussian} into \eqref{eq:Wbc-WG} yields
\begin{align}
    W_n^{bc} = \tilde W_n + O_{a.s.}\left((nh)^{-1/6}\log n + n^{1/2} h^{3/2} + h \log h^{-1} + h^{-1/2}n^{-1/4}\ell(n) \right). \label{eq:Wbc-Gaussian}
\end{align}
\cref{ass:bandwidth} implies that $n^{1/2} h^{3/2} = o(\log^{-1/2}(n))$ and $h^{-1/2}n^{-1/4}\ell(n) = o(\log^{-1/2}(n))$. Therefore,
\begin{align}
    W_n^{bc}-\tilde W_n=o_p(\log^{-1/2} n).    
\end{align}
It now follows from \citet[][Remark 3.2]{chernozhukov2014gaussian} that
\begin{align}
    \sup_{t\in \R} \left|\Pb(W_n^{bc}\le t)-\Pb(\tilde W_n\le t)\right| \to 0. \label{E:Kolm-conv-2}
\end{align}
Applying the triangle inequality to equations \eqref{E:Kolm-conv-1} and \eqref{E:Kolm-conv-2} yields 
\begin{align}
    \sup_{t\in \R} \left|\Pb(W_n^{bc}\le t)-\Pb(W_n^\G\le t)\right| \to 0. \label{E:Kolm-conv-3}
\end{align}
On the other hand, considering the sample $U_i=F(X_i) \sim \text{iid Uniform}[0,1]$, we have 
\begin{align}
    W_n^{bc,U[0,1]} = W_n^\G +  O_{a.s.}\left((nh)^{-1/6}\log n + n^{1/2} h^{3/2} + h \log h^{-1} + h^{-1/2}n^{-1/4}\ell(n) \right).
\end{align}
A similar argument yields
\begin{align}
    \sup_{t\in \R} \left|\Pb(W_n^{bc}\le t)-\Pb(W_n^{bc,U[0,1]}\le t)\right| \to 0,
\end{align}
which completes the proof. $\qed$

\end{document}